\documentclass[twocolumn,showpacs,aps,pra,amsfonts,amsmath,amssymb,floatfix]
{revtex4-1}

\usepackage{subfigure}
\usepackage{graphicx, amsfonts, amsthm, amsmath, amssymb, amstext, latexsym, float,color}
\usepackage{comment}
\usepackage{fullpage}

\newcommand{\ave}[1]{ \left\langle #1  \right\rangle}
\newcommand{\eat}[1]{}
\def \tr{{\textrm {Tr}}}

\providecommand{\U}[1]{\protect\rule{.1in}{.1in}}
\newtheorem{theorem}{Theorem}
\newtheorem*{theorem*}{Theorem}

\newtheorem{proposition}[theorem]{Proposition}
\newtheorem{remark}[theorem]{Remark}

\begin{document}
\title{Could Gaussian regenerative stations act as quantum repeaters?}

\author{Ryo Namiki$^1$, Oleg Gittsovich$^{1}$, Saikat Guha$^2$, and Norbert L{\" u}tkenhaus$^1$}
\affiliation{$^1$Institute for Quantum Computing and Department of Physics and Astronomy, University of Waterloo, Waterloo ON, Canada N2L 3G1\\
$^2$Quantum Information Processing group, Raytheon BBN Technologies, Cambridge MA, USA 02138}

\begin{abstract} 

Higher transmission loss diminishes the performance of optical communication---be it the rate at which classical or quantum data can be sent reliably, or the secure key generation rate of quantum key distribution (QKD). Loss compounds with distance---exponentially in an optical fiber, and inverse-square with distance for a free-space channel. In order to boost classical communication rates over long distances, it is customary to introduce regenerative relays at intermediate points along the channel. It is therefore natural to speculate whether untended regenerative stations, such as phase-insensitive or phase-sensitive optical amplifiers, could serve as repeaters for long-distance QKD. The primary result of this paper 　rules out all  bosonic Gaussian channels 
 to be useful as QKD repeaters, which include phase-insensitive and phase-sensitive amplifiers as special cases, for any QKD protocol. We also delineate the conditions under which a Gaussian relay renders a lossy channel {\em entanglement breaking}, which in turn makes the channel useless for QKD.
\end{abstract}

\maketitle

\section{Introduction}

In recent years, performance of various communication tasks over an optical channel---when limited only by the fundamental noise of quantum mechanical origin---have been extensively studied. A few examples are: finding the communication capacities of the lossy optical channel for transmitting classical information~\cite{Gio04}, quantum information~\cite{Wol98}, and that for transmitting both classical and quantum information simultaneously in the presence of a limited amount of pre-shared entanglement~\cite{Wil12}. One of the biggest breakthroughs in optical communication using quantum effects was the invention of quantum key distribution (QKD), which is a suite of protocols that can generate information-theoretically-secure shared secret keys~\cite{Scarani_Renner_2008} between two distant parties Alice and Bob over a lossy-noisy optical channel, with the assistance of a two-way authenticated public classical channel. Security of QKD leverages quantum properties of light to ensure the generated shared keys are secure from the most powerful adversary that is physically consistent with the channel noise collectively estimated by Alice and Bob (despite the fact that much of that noise may actually be caused by non-adversarial or natural causes). Various QKD protocols have been proposed in the last three decades~\cite{SBCDLP09}, some of which have been transitioning to practice~\cite{DARPA02,SECOQC09,TOKYO_QKD}.

For all the communication protocols discussed above, the rates decrease rapidly with channel loss.  For the task of classical communication over an ideal pure-loss channel (modeled by a beamsplitter of transmittance $\eta$), at any given value of the channel transmittance $\eta$, no matter how small, the data rate can in principle be increased without bound by increasing the input power~\cite{footnote1}. For QKD, this is not the case. For several well-known QKD protocols (such as BB84~\cite{BB84} with single photons and BB84 with weak laser light encoding and decoy states, E91~\cite{E91} with an ideal entanglement source, and CV-QKD with Gaussian modulation~\cite{Sil02, Gro03}), the secret key rate $R$ decays linearly with channel transmittance $\eta$ in the high-loss ($\eta \ll 1$) regime~\cite{footnote}. Recently, it was shown that this linear rate-transmittance scaling over the lossy bosonic channel---for secure-key generation with two-way public classical communication assistance---is  impossible to improve upon, no matter how one may design a QKD protocol, or how much input power is used~\cite{Tak13}. To be specific, the secret key rate of any QKD protocol must be upper bounded by $R_{\rm UB}$ measured in bits/mode and given by 
\begin{equation}
\label{TGWbound}
R_{\rm UB} = \log_2\frac{1+\eta}{1-\eta} ,
\end{equation} which equals $R_{\rm UB} \approx 2.88\eta$, for $\eta \ll 1$. This fundamental rate-loss upper bound also applies to the following related tasks: quantum communication (sending qubits noiselessly over a lossy channel), direct secure communication~\cite{footnote3}, and entanglement generation (where each task may also use assistance of a separate authenticated two-way classical communication channel, in addition to transmissions over the quantum lossy bosonic channel itself)~\cite{Tak13}.

As we discussed above, for classical communication over an ideal lossy channel, one could in principle increase the input power without bound as the loss increases, to maintain a required data rate. However, an unbounded input power is impractical both from the point of view of the availability of a laser that is powerful enough, and also to avoid hitting up against the fiber's non-linearity-driven peak power constraint. This is why traditionally, electrical regenerators have been used to compensate for loss in long-haul optical fiber communications, which help restore the signal-to-noise ratio (SNR) of the digitally-modulated signals by periodically detecting and regenerating clean optical pulses. Over the last few decades, all-optical amplifiers, such as erbium-doped fiber amplifiers (EDFAs), have become popular in lieu of electrical regenerators, both due to their greater speeds as well as the low noise of modern EDFAs. Caves analyzed the fundamental quantum limits on the noise performance of optical amplifiers~\cite{Cav82}, for both phase-insensitive (PIA) and phase-sensitive amplifiers (PSA). Loudon analyzed the fundamental limitations on the overall SNR to `chains' of loss segments and optical amplifiers, both in the context of phase-sensitive (coherent detection) receivers, as well as direct detection receivers~\cite{Lou85}.

For QKD, one way to beat the linear rate-transmittance scaling  is to break up the channel into low-loss segments by introducing physically-secured center stations; in this approach the overall key rate is still upper bounded by $R \le \log_2[(1+\eta^\prime)/(1-\eta^\prime)]$ bits/mode, but $\eta^\prime$ is the transmittance of the longest (lossiest) segment. Quantum repeaters are conceptual devices~\cite{SST_2011, LST_2009}, which if supplied at these intermediate stations, can beat the linear rate-transmission scaling without having to physically secure them.  There is an approach  to build a quantum repeater using one-way communication only~\cite{Mur14}, so they can act as passive untended devices. However, such structured implementations of those devices require quantum error correction codes operating on blocks of multiple qubits.  A recently-proposed repeater protocol~\cite{Azu13} even eliminates the requirement of a quantum memory, but utilizes photonic cluster states.  Building a functional quantum repeater  is subject to intensive fundamental research, but is currently far from being a deployable technology. The natural question that thus arises---in analogy to Loudon's setup for classical optical communication~\cite{Lou85}---is whether all-optical amplifiers (PIAs or PSAs), left untended and inserted at regular intervals, might act to some degree as quantum repeaters and thereby help boost the distances over which QKD can be performed over a lossy channel.

{
The remainder of the paper is organized as follows.  First, in Section~\ref{sec:MainResults} we summarize the main results derived in this article to put it into perspective. A central finding is a decomposition of a lossy quantum channel with intermediate  bosonic Gaussian channel stations into another form without any insertion of middle stations as depicted in Fig.~\ref{fig:summary}.  We then continue into the technical part.
 In Section~\ref{sec:Gaussian}, we give an overview of bosonic Gaussian states and channels. In Section~\ref{sec:GaussianStations}, we analyze the scenario when a general multi-mode Gaussian channel is inserted between two pure-loss segments, and show how one can collect the entire pure loss in the center of the channel by appropriate modifications to the transmitter and the receiver. In Section~\ref{sec:EBconditions}, we consider single-mode Gaussian stations, and delineate the conditions for when the Gaussian center station renders the concatenation with the losses on its two sides, an entanglement-breaking channel. The quantum limited stations, the PSA and the PIA, are addressed as special cases. We conclude in Section~\ref{sec:conclusions} with a summary of the main results, and thoughts for future work.
} 
\begin{figure}[htb]
\centering
\includegraphics[width=\columnwidth]{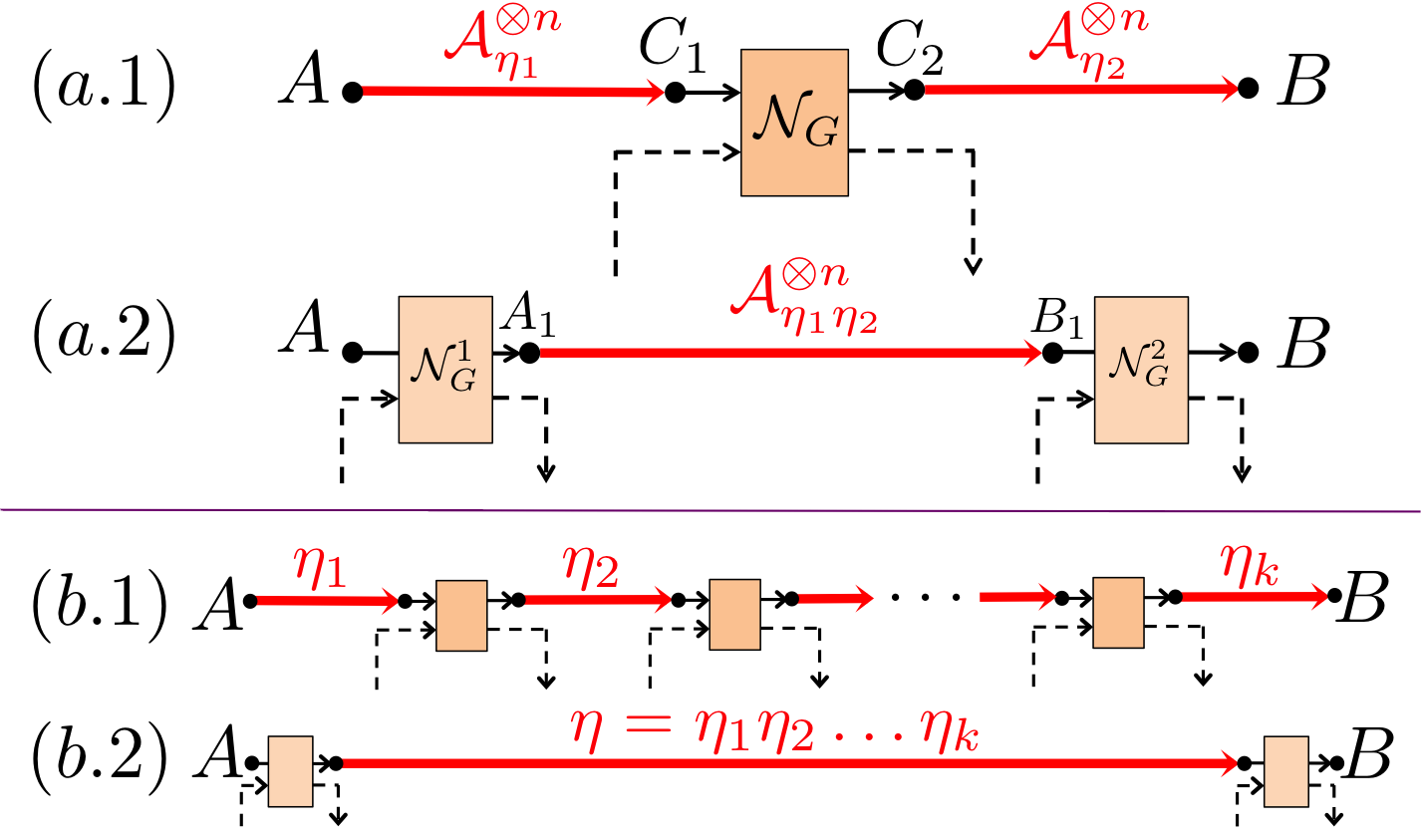} 
\caption{(a) Any $n$-mode Gaussian channel ${\cal N}_G$ sandwiched between two pure-loss channel segments ${\cal A}_{\eta_1}^{\otimes n}$ and ${\cal A}_{\eta_2}^{\otimes n}$, respectively, can be decomposed into  a single lossy channel ${\cal A}_{\eta}^{\otimes n}$ sandwiched by a pair of Gaussian channels, ${\cal N}_G^1$ and ${\cal N}_G^2$. The net loss in the channel is the sum (in dB) of the losses of the two individual lossy segments, i.e., $\eta = \eta_1\eta_2$ and the Gaussian channel at the receiver end ${\cal N}_G^2$ is a Gaussian unitary map. (b) Using this transformation recursively, one can `push' a collection of general Gaussian center stations interspersed through a lossy channel (b.1) to a single Gaussian operation at the input, and a single Gaussian operation at the output (b.2) of the entire loss accumulated in the center.}
\label{fig:summary}
\end{figure}

\section{Outline of main results} \label{sec:MainResults}

In this paper, we show that such is not possible when those all-optical amplifiers are limited to 
Gaussian channels. Note that by using the word ``channel'' we automatically imply the action to be trace-preserving. Examples of such channels involve beamsplitters, phase-shifters and squeezers~\cite{Bra05}. We prove our claim  by transforming a concatenation of two lossy channel segments with a Gaussian channel in the middle, as a pair of Gaussian channels at the two ends, with the total loss collected in the middle (see Fig.~\ref{fig:summary}). The implication of our no-go result is that if  Gaussian channels are employed in center station(s) placed along a lossy channel, the overall QKD key rate, for any QKD protocol, must be upper bounded by $R_{\rm UB} = \log_2[(1+\eta)/(1-\eta)]$ bits/mode, with $\eta$ being the total end-to-end channel transmittance.  
 {Simple protocols such as laser-decoy-based BB84, or Gaussian-modulated laser-based CV protocols, with no repeaters, can already attain key rates that have the optimal (linear) rate-transmittance scaling and are only a small constant factor below the general upper bound~\cite{SBCDLP09}.}

 For any optical communication protocol over a lossy channel interspersed with Gaussian stations, our result shows there exists {\em another} protocol with the same performance that does not use any intermediate station, which can be derived from the original protocol by suitably amending the transmitted signals and the receiver measurement.   Our result does {\em not} preclude a Gaussian channel in the middle of the lossy channel to improve the performance of a {\em given} protocol, {\em if} the transmitter and receiver are held to be the same. Nor does it preclude the existence of scenarios where it might be technologically easier to implement a protocol with such intermediate stations, as opposed to modifying the transmitter and the receiver per the prescription generated by our analysis. An example of such improvement is the increased range of a QKD protocol with a given level of detector noise (although, any increase in range must be consistent with the $R \sim \eta$ rate-transmission scaling).

{
 Given that the overall rate-transmission scaling can not be changed, the question remains whether there might be other implementation advantages of Gaussian center stations. It turns out that there are strict conditions on such a scenario. To demonstrate this, }we delineate the conditions under which a Gaussian center station causes a lossy channel to become entanglement breaking (EB) \cite{Horo03,Hol07,Hol08}. It is well known that QKD is not possible on an EB channel, since the output of an EB channel can be simulated quantitatively correctly using a measure-and-prepare scheme~\cite{Cur04}. The pure lossy channel is not EB by itself for any non-zero transmittance, $\eta >0$. 

Let us illustrate our reasoning for the better known case of classical communication over pure-loss bosonic channels. The channel capacity of the lossy bosonic channel (described by single photon transmittance $\eta$) using signals with mean photon number $\bar{n}$ per mode, is given by  $g(\eta \; {\bar n}) = (1+\eta \; {\bar n})\log_2(1+\eta\; {\bar n}) - \eta \; {\bar n}\log_2 \eta \;{\bar n}$ bits per mode~\cite{Gio04}. We see, that increasing the mean photon number increases the classical communication rate. In practice, it is impractical to keep increasing the mean photon number due to non-linear effects in the fiber that limit the input power and can distort the signals. For these reasons, one limits the input power and builds optical amplifiers (phase sensitive or phase insensitive) into the fiber. According to our theorems, for the ideal loss-only bosonic channel, the setup of lossy segments with intermediate amplifiers is equivalent to a new transmitter consisting of the old transmitter combined with a very strong amplifier, followed by a transfer through the full distance of the lossy bosonic channel, and then a  receiver consisting of a combination of another amplifier and 
on the original
 receiver. 
 This replacement protocol corresponds to the situation of using a large input mean photon number, and realizes the classical capacity of the lossy bosonic channel. What we learn is that the intermediate amplifiers do not increase the channel capacity of the lossy bosonic channel, but realize an equivalent protocol that keeps the optical signals---throughout the communication channel---within a peak power level that is sufficiently below the level where non-linear effects would be encountered.

In QKD, the secrecy capacity of the lossy bosonic channel does not increase unboundedly with the input power of the signals, thus using strong signal pulses pushing into the non-linear domain of fibers is not important for QKD protocols: the use of equivalent replacement schemes utilizing optical amplifiers would not give any advantage. To the contrary, amplifiers will add additional noise which will eventually be detrimental to the performance of the QKD protocol, with the exception of effects of noisy pre-processing that can increase the secret key rate compared to protocols not using this approach~\cite{renner05b}. Note however that noisy pre-processing cannot improve on the fundamental secrecy capacity of the lossy bosonic channel, which is solely a function of the channel's end-to-end loss.

Our main result adds to the list of no-go results for Gaussian operations in quantum information protocols, i.e., those that cannot be performed with Gaussian operations and classical processing alone. Some examples are universal quantum computing~\cite{Bartlett2002}, entanglement distillation of Gaussian states~\cite{Eisert2002,Fiurasek2002,Giedke2002}, optimal cloning of coherent states~\cite{Cerf2005}, optimal discrimination of coherent states~\cite{Takeoka2008,Tsujino2011,Wittmann2010-1,Wittmann2010-2}, Gaussian quantum error correction~\cite{Niset2009}, and building a joint-detection receiver for classical communication~\cite{Tak14}.

\section{Gaussian states and channels}
\label{sec:Gaussian}
In this section, we will provide a basic introduction to the mathematics of Gaussian states and channels, sufficient to develop the results in this paper. For a more detailed account, see Ref.~\cite{Giedke2002}.  A quantum state $\rho$ of an $n$-mode bosonic system is uniquely described by its characteristic function
\begin{equation}
\label{eq:characteristic_function}
\chi (\mu) = {\rm Tr} \left[ \rho \mathcal{W} (\mu) \right] ,
\end{equation}
where the Weyl operator, $\mathcal{W} (\mu) = \exp \left[ - i \mu^T R \right]$, with $R=[ \hat{x}_1, \cdots , \hat{x}_n , \hat{p}_1, \cdots , \hat{p}_n]^T$ consisting of field quadrature operators of the $n$ modes satisfying the commutation relations $[\hat{x}_k, \hat{p}_l]=i \delta_{kl}$, with $\mu =[ \mu_1, \cdots , \mu_{2n} ]$ a $2n$-length real vector. The characteristic function of a Gaussian state $\rho$ is given by, 
\begin{equation}
\label{eq:chi_gaussian}
\chi_\rho (\mu) = \exp\left[ - \frac{1}{4} \mu^T \gamma \mu + i d^T \mu \right] ,
\end{equation}
where the $2n \times 2n$ matrix $\gamma$ is the covariance matrix (CM) and the $2n$-length vector $d:= (\ave{\hat x} , \ave{\hat p}  )^T$ is the mean, or the displacement vector (DV), of $\rho$.   The Gaussian state $\rho$ can thus be described uniquely by the pair $(\gamma, d)$.
Due to the canonical uncertainty relation, any CM of  physical states  has to satisfy  
\begin{eqnarray}
\gamma \ge \frac{i}{2} \sigma ,   \label{PhysConState} 
\end{eqnarray} where  \begin{equation}
\label{eq:S_decomposition}
\sigma : =  \left( 
\begin{array}{cc}
0  & \openone _ n \\
-\openone _ n & 0
\end{array}
\right).
\end{equation}
 A Gaussian unitary operation $U_G$ transforms a Gaussian state $(\gamma, d)$ to a Gaussian state $(\gamma^\prime, d^\prime)$  as  
\begin{equation}
\gamma^\prime = M^T \gamma M, \quad d^\prime = M^T d, 
\label{eq:symplectic}
\end{equation}
where $M$ is a symplectic matrix that satisfies   \begin{eqnarray}
M ^T  \sigma M = \sigma   . \label{Guni} 
\end{eqnarray} 
A Gaussian channel ${\cal E}$ can be described by a triplet $(K,m,\alpha)$~\cite{Hol08}. It transforms a state $(\gamma, d)$ to the state $(\gamma^\prime, d^\prime)$ as 
\begin{eqnarray}
\gamma^\prime = K^T \gamma K + \alpha, \quad d^\prime = K^T d +m. 
\end{eqnarray}
From the regularity of CMs in Eq.~\eqref{PhysConState}  the physical condition for the pair $(K,\alpha)$ is given by 
\begin{eqnarray}
\alpha \ge \frac{i}{2}(\sigma - K^T \sigma K )  \label{PhysCon}. 
\end{eqnarray}
Composition of two Gaussian channels $\mathcal E_1$ and $\mathcal E_2$ yields another Gaussian channel $\mathcal E_{12} =  \mathcal E_2 \circ \mathcal E_1$, where
\begin{eqnarray}
K_{12} &=& K_1 K_2 \nonumber, \\
m_{12} &=& K_2^T m_1 + m_2,  \,{\text{and}}\nonumber \\  
\alpha_{12} &=& K_2^T \alpha _1 K_2 + \alpha_2. \label{CompoRule}
\end{eqnarray}
In this paper, we will focus on Gaussian channels with $m=0$.  
 In Appendix~\ref{app:displacement} we show an explicit calculation demonstrating how   mean displacement terms can be separated out in any concatenation of Gaussian channels.  

 In the following subsections, we will delve a little deeper into properties of single-mode Gaussian  channels that we use later on.

\subsection{Decomposing a Gaussian unitary operation}
The symplectic matrix $M$ in Eq.~\eqref{eq:symplectic} of a Gaussian unitary can always be decomposed as 
\begin{equation}
\label{eq:S_decomposition}
M = B \left( 
\begin{array}{cc}
\Lambda & 0 \\
0 & \Lambda^{-1}
\end{array}
\right) B' ,
\end{equation}
where $\Lambda$ is a positive diagonal matrix, and $B$, $B'$ are orthogonal symplectic matrices ($B^T = B^{-1}$)~\cite{Bra05}. This implies that any $n$-mode Gaussian unitary operation $U_G$ can be realized by a passive linear optic circuit $B$ (a circuit involving only beamsplitters and phase-shifters~\cite{Rec94}), followed by $n$ parallel (tensor-product) single-mode squeezers, followed by another $n$-mode passive linear optic circuit $B'$~\cite{Bra05}. Therefore a general Gaussian unitary operation can always be decomposed into passive linear optics (beamsplitters and phase shifters), single-mode squeezing and single-mode displacement operations.

Therefore, up to a displacement, a single-mode Gaussian unitary (described by its symplectic matrix $M$) can be decomposed as 
\begin{eqnarray}
M  =  R_\theta  S_{G}  R_\phi,
\end{eqnarray}
where 
\begin{eqnarray}
R_\theta &=& \left(
  \begin{array}{cc}
 \cos \theta   &  \sin \theta    \\
  -\sin\theta  & \cos \theta   \\
  \end{array}
\right)
\end{eqnarray}
is the symplectic matrix of a single-mode phase rotation, and 
\begin{eqnarray}
S_{G} &=& \left(
  \begin{array}{cc}
    \sqrt G  +\sqrt{G -1}   & 0   \\
       0 &   \sqrt G  -\sqrt{G -1}   \\
  \end{array}
\right)  \end{eqnarray} 
is the symplectic matrix of a single-mode (phase-quadrature) squeezer. Note that it is sufficient to restrict the above decomposition to a phase-quadrature squeezer, is because one can absorb any additional phase in the squeezing operation into $R_\theta$ and $R_\phi$. This is because the symplectic matrix of a single-mode squeezer with gain $G$ and squeezing angle $\theta'$ can be expressed as 
\begin{eqnarray}
S_{G, \theta'} =  R_{\theta'} S_{G} R_{\theta'} ^\dagger.
\end{eqnarray} 

\subsection{Entanglement breaking channels} 

An {\em entanglement breaking} (EB) channel is one whose action on one half of an entangled state (with an identity map on the other half) always yields a {\em separable} state. An EB channel  can always be written in a {\em measure-and-prepare} form~\cite{Horo03, Hol08}. (See also Eq.~(\ref{eq:measureprepare}) below.) Any concatenation of $n$ (not necessarily Gaussian) channels, ${\cal E}_n \circ \ldots \circ {\cal E}_2 \circ {\cal E}_1$ is EB if one of channels ${\cal E}_i$ is EB. It is instructive to see the argument explicitly for $n= 3$. Consider the serially-concatenated channel, ${\cal E}_t= {\cal E}_3  \circ {\cal E}_2 \circ {\cal E}_1$, where the center station ${\cal E}_2 $ is EB. Supposing its measure-and-prepare form is given by  ${\cal E}_2  (\rho ) = \sum_k  \tr(M_k \rho )\sigma_k$,  with $M_k \ge 0$ and $\sigma_k \ge 0 $, we can write ${\cal E}_t$ in a measure-and-prepare form,
\begin{equation}
\label{eq:measureprepare}
{\cal E}_t  (\rho) =  \sum_k  \tr[M_k {\mathcal E}_1 (\rho) ] {\mathcal E}_3(\sigma_k) =   \sum_k  \tr[M'_k \rho ] \sigma_k ',
\end{equation}
where $ \sigma_k '= {\mathcal E}_3(\sigma_k) \ge 0 $, and it is straightforward to show that $M'_k = \sum_i A_i^\dagger M _k A_i \ge 0$, where $\{A_i\}$ represent Kraus operators of $\mathcal E_1$, i.e., ${\mathcal E}_1(\rho) = \sum_i A_i \rho A_i^\dagger$.

The measure-and-prepare representation of an EB channel implies that the channel's quantum transmission can be seen as transmission of the (probabilistic) classical information obtained as a result of a hard quantum measurement made on the channel's input. This is the intuition behind why such a channel has zero secret-key capacity, and thus cannot be useful for QKD~\cite{Cur04}. Because of this reason, when we analyze concatenations of several Gaussian center stations for potential use as repeaters, we will limit our discussion to the case when all the channels $\mathcal E_i$ in the concatenation are non EB (since this is a necessary condition for QKD). Note however that when interspersed with loss segments, even when all center stations are non-EB, the overall input-output map can become EB---a topic that we will discuss in more detail later in Section~\ref{sec:EBconditions}.

\subsection{Unitary-equivalence classification for single-mode Gaussian channels}\label{sec:classification}
Our analysis of general one-mode Gaussian operations will be based on the standard forms of such operations obtained from the unitary equivalence classification of quantum channels developed by Holevo~\cite{Hol07, Hol08}. We say that two quantum channels $\Phi$ and $\Phi_S$ are \textit{unitary equivalent} if there exist unitary operators $U_V, U_W$ such that,
\begin{eqnarray}
\Phi_S (\rho) = U_W \Phi  (U_V \rho U_V ^\dagger) U_W^\dagger. 
\label{eq:unitary_equivalence}
\end{eqnarray}
If $U_V$ and $U_W$ above are Gaussian, we say $\Phi$ and $\Phi_S$ are {\em Gaussian unitary equivalent}. If a single-mode Gaussian channel ${\cal E} \triangleq (K, m, \alpha)$ is not an EB channel, it must be Gaussian unitary equivalent to a channel belonging to one of the following two classes:  \\

\noindent {\bf (i) Phase insensitive channel (PIC)}: This class of channels is described by the triplet $(K, 0, \alpha)$, with
\begin{eqnarray}
K &=& \sqrt \kappa \openone_2, \,{\text{and}}\nonumber \\
\alpha &=&(|1-\kappa|/2 + N )\openone_2 , \label{i}
\end{eqnarray} 
where $N \ge 0$ is the excess noise parameter and $\kappa \ge 0$ is a gain parameter. We will denote 
 this channel as ${\cal A}_\kappa^N$. It acts on the canonical quadratures phase-insensitively. When the gain $\kappa \ge 1$, we call it the {\em phase-insensitive amplifier} (PIA). When $\kappa < 1$, we call it the lossy bosonic channel (with excess thermal noise $N$). In this case, $\kappa$ is the channel's transmittance, the fraction of the input photons that appear at the channel's output. We will use the shorthand notation, ${\cal A}_\kappa \equiv {\cal A}_\kappa^0$ for a quantum-limited phase-insensitive amplifier, or a pure-loss channel, for $\kappa \ge 1$ and $\kappa < 1$, respectively. 

It is known that the PIC is EB if and only if~\cite{Hol08}, 
\begin{eqnarray}
N  \ge \min (1, \kappa). 
\label{38}
\end{eqnarray}
In our analysis we will assume that the PIC is not EB, i.e., $N \in [0, \min(1, \kappa))$. Furthermore, using the composition rule of Eq. (\ref{CompoRule}), it is easy to see that any single-mode rotation (unitary) ${\cal R}$ commutes with a PIC, i.e.,
${\cal R} \circ {\cal A}_\kappa^N =  {\cal A}_\kappa^N \circ {\cal R}$.
\\

\noindent {\bf (ii) Additive noise channel (ANC)}: This is a class of {\em phase-sensitive} Gaussian channels that adds rank-1 noise to the input state, and is described by the triplet $(K, 0, \alpha)$, with 
\begin{eqnarray}
K &=&  \openone_2, \,{\text{and}} \nonumber \\
\alpha&= &  \frac{1}{2}\textrm{diag} (0,\epsilon),  \label{ii}
\end{eqnarray} 
where the noise parameter $\epsilon > 0$. We will denote this channel as $\mathcal I^{\epsilon}$, and will call it the additive noise channel (ANC).

\section{Gaussian regenerative stations in a lossy channel}\label{sec:GaussianStations}
In this section we investigate lossy bosonic channels that have 
intermediate Gaussian channels 
inserted at some intervals. We will show that such an arrangement is still equivalent (up to Gaussian operations at the entrance and the exit) to a lossy bosonic channel with the total loss of the original loss segments. As a consequence, insertion of Gaussian channels 
 cannot increase the secrecy capacity of the lossy bosonic channel. 

The setup for the main result of this paper is schematically depicted in Fig.~\ref{fig:summary}. Consider a pure-loss optical channel ${\cal A}_\eta$ with a given amount of total end-to-end ($A$ to $B$) transmittance $\eta \in (0, 1]$. Let us place a Gaussian center station---a quantum channel, or a trace-preserving completely positive map, ${\cal N}_G^{C_1 \to C_2}$---somewhere in the middle, thereby splitting ${\cal A}_\eta$ into two pure-loss segments: a pure-loss channel with transmittance $\eta_1$, ${\cal A}_{\eta_1}$ ($A$ to $C_1$), and a pure-loss channel with transmittance $\eta_2$, ${\cal A}_{\eta_2}$ ($C_2$ to $B$), such that $\eta_1\eta_2 = \eta$. We show that the overall channel action from $A$ to $B$ is unaffected by the transformation shown  in Fig.~\ref{fig:summary}(a), which replaces the Gaussian center station ${\cal N}_G^{C_1 \to C_2}$ by a Gaussian operation ${{\cal N}_G^1}^{A \to A_1}$ at the input of the channel and a Gaussian operation ${{\cal N}_G^2}^{B_1 \to B}$ at the output of the channel. By applying this transformation recursively, it is easy to see that one can replace any number of Gaussian center stations interspersed through the lossy channel ${\cal A}_\eta$ into two Gaussian operations, at the input and the output, respectively.

Let us consider an $n$-mode lossy bosonic channel ${\cal A}_\eta ^{\otimes n } \triangleq (K_0, 0,\alpha_0) $ with  
\begin{eqnarray}
K_0  &=&  \sqrt{\eta } \openone_{2n} \nonumber, \\
\alpha_0  &=& \frac{1- \eta}{2} \openone_{2n}. 
\end{eqnarray}
Let ${\cal N}_G  \triangleq (K, 0, \alpha)$ denote an $n$-mode Gaussian channel, which we consider as the candidate for a center station. Note that this Gaussian center station could act collectively on $n$ spatial and/or temporal modes of the propagating field. The main result of this section is the proof of the following proposition, also depicted schematically in Fig.~\ref{fig:summary}.

\begin{proposition}\label{prop:n-mode}
For any $n$-mode Gaussian channel ${\cal N}_G$ there exists a Gaussian channel  ${\cal N}_G^1$  and a Gaussian unitary channel  ${\cal N}_G^2$   that satisfy
\begin{eqnarray}
{\cal A}_{\eta_2} ^{\otimes n } \circ{\cal N}_G  \circ  {\cal A}_{\eta_1} ^{\otimes n } = {\cal N}_G^2 \circ {\cal A}_{\eta_1 \eta_2} ^{\otimes n } \circ{\cal N}_G^1.   
\end{eqnarray}
\end{proposition}

\begin{proof}
Our goal is to find a pair of  Gaussian channels ${\cal N}_G^1$ and  ${\cal N}_G^2$ that satisfies the physical condition Eq. (\ref{PhysCon}). 
From the composition rule of Eq.~\eqref{CompoRule}  we find the total channel action   $\Phi_t :=  {\cal A}_{\eta_2} ^{\otimes n } \circ{\cal N}_G  \circ  {\cal A}_{\eta_1} ^{\otimes n}$  can be described by $\Phi_t  \triangleq (K_t, 0, \alpha_t)$ with  
\begin{eqnarray}
K_t  &=&  \sqrt{\eta_1 \eta_2}K \nonumber, \\
\alpha_t  &=& \eta_2\left(\frac{1- \eta_1}{2} K^T K + \alpha    \right)+  \frac{1- \eta_2}{2} \openone_{2n} \label{DefAlpha}. 
\end{eqnarray}

We will prove the proposition by constructing the required Gaussian channels using a symplectic matrix  denoted by $M$. The properties of this matrix and its existence
 are the subject of the following theorem:
\begin{theorem}\label{thm1}
For a given $\alpha_t$ in Eq.~\eqref{DefAlpha}, there exists a CM matrix $\gamma^\prime$ and a symplectic matrix $M$  such that 
\begin{equation}
\alpha_t = \eta_1 \eta_2 \alpha + (1-\eta_1 \eta_2) \gamma^\prime 
\end{equation}
and 
\begin{equation}
M^T \gamma ^\prime M \ge  \frac{1}{2}  \openone_{2n}.
\end{equation}
\end{theorem}
\begin{proof}
From the physical condition of a Gaussian channel in Eq.~\eqref{PhysCon}, we have
\begin{eqnarray}
 \frac{i \sigma }{2}& \le&  \frac{i \sigma}{2} + \frac{1}{2} K^T \left( \openone_{2n} - i \sigma \right) K \\
& = & \frac{1}{2} \left( K^T K + i  (\sigma - K^T \sigma K  )\right) \\
& \le&     \frac{1}{2} ( K^T K + 2 \alpha  ) 
 \end{eqnarray}
where we used  in the first line that the matrix $ \openone_{2n} - i \sigma$ is positive semi-definite and in the last line that ${\cal N}_G$ is a physical channel. Our calculation implies $ \gamma := \frac{1}{2} ( K^T K + 2 \alpha  ) $ is a CM of an $n$-mode Gaussian state due to Eq.~\eqref{PhysConState}.
Consider now the convex combination of this CM with the CM of the $n$-mode vacuum state  $\gamma^\prime :=  p \; \gamma  + (1- p)\; \frac{1}{2} \openone_{2n}$with mixing probability $p =   \eta_2(1-\eta_1 ) /(1-\eta_1 \eta_2 ) \in [0,1]$.  A straightforward calculation verifies that $\alpha_t = \eta_1 \eta_2 \alpha + (1-\eta_1 \eta_2) \gamma^\prime $. As $\gamma^\prime$ is a valid CM,  there exists a symplectic matrix $M$ such that  one obtains a diagonal form $M^T \gamma ^\prime M \ge  \frac{1}{2}  \openone_{2n}$, which corresponds to a product of thermal states. 
\end{proof}

We are now in a position to define the Gaussian channels ${\cal N}_G^1\triangleq (\tilde K, 0,  \tilde \alpha)$ and  ${\cal N}_G^2\triangleq (M^{-1}, 0,  0)$ with the help of  
\begin{eqnarray}
\tilde K & = & \frac{1}{\sqrt{\eta_1 \eta_2}} K_t M \;\;\;\left( \equiv K M\right) \; ,\\
\tilde \alpha &=& \frac{1}{\eta_1 \eta_2 }  \left(M^T \alpha _t M -   \frac{1- \eta_1 \eta_2}{2} \openone_{2n}\right).  \label{13siki}
\end{eqnarray}

To show that the channels are proper physical channels, we can concentrate on  ${\cal N}_G^1$ since ${\cal N}_G^2$ corresponds to a unitary Gaussian channel.  To prove that ${\cal N}_G^1$ is physical, we use in a first step the results of theorem \ref{thm1}, and then the physicality constraints on the channel ${\cal N}_G$, followed by a rewriting of the variables. These steps allow us to obtain
\begin{equation}
\tilde \alpha   \ge M^T \alpha M \ge M^T \frac{i}{2}(\sigma - K^T \sigma K) M = \frac{i}{2}(\sigma - {\tilde K}^T \sigma  \tilde K),  
\end{equation} 
Hence,  ${\cal N}_G^1 \triangleq (\tilde K, 0,  \tilde \alpha) $ is a valid Gaussian channel.   It is again straightforward to verify that  $\Phi _t =  {\cal N}_G ^2 \circ     {\cal A}_{\eta_1 \eta_2} ^{\otimes n } \circ{\cal N}_G ^1  $. This  proves  proposition \ref{prop:n-mode}.  
 \end{proof}
 Overall, we showed the equivalence of a 
 bosonic Gaussian channel sandwiched between two lossy bosonic channels to a single lossy bosonic channel, bearing the total loss of the the original bosonic channels, and now sandwiched between two Gaussian channels. 
 This corresponds to the conversion of $(a.1)$ into $(a.2)$ of Fig.~\ref{fig:summary}. A simple iteration of this result shows that any pattern of Gaussian channels interspersed 
  between loss segments can be rearranged into a lossy bosonic channel sandwiched between Gaussian channels [See   Fig.~\ref{fig:summary}(b)].

As the initial Gaussian channel can be combined with the state preparation, and the final Gaussian channel can be combined with the detection setup, it is evident that the total secret key rate of this arrangement is still bound by  $R_{\textrm{UB}}$ of  Eq.~(\ref{TGWbound}).

\section{Entanglement-breaking conditions for single-mode center stations}\label{sec:EBconditions}
 As discussed in the introduction, there might be practical reasons 
  one want to use interspersed intermediate stations, even if the resulting key rate is still limited by the bound of Eq.~(\ref{TGWbound}). In this sections we will demonstrate severe restrictions on the situations where such an advantage may exist. To do so, we will investigate when such a sequence of lossy channels and Gaussian center stations becomes entanglement breaking (EB) so that its secrecy capacity goes to zero \cite{Cur04}. 
In the following part of this article we execute the central first step of such an investigation and focus on 
 single-mode Gaussian channels. 

A pure-loss channel is not EB by itself, but increasing loss could make the channel progressively more fragile and susceptible to being EB when concatenated with other Gaussian operations, such as amplifiers.  In the following subsections, we show the explicit conditions on the parameters of a  Gaussian non-EB center station ${\cal N}_G$, such that the composition $\Phi_0 \equiv {\cal A}_{\eta_2} \circ {\cal N}_G \circ {\cal A}_{\eta_1}$ is EB, and specialize the conditions to the cases when ${\cal N}_G$ is either a PSA or a PIA. 

\subsection{General non-EB center stations}\label{sec:general_nonEB}
\begin{figure*}
\centering
\includegraphics[width=0.85\textwidth]{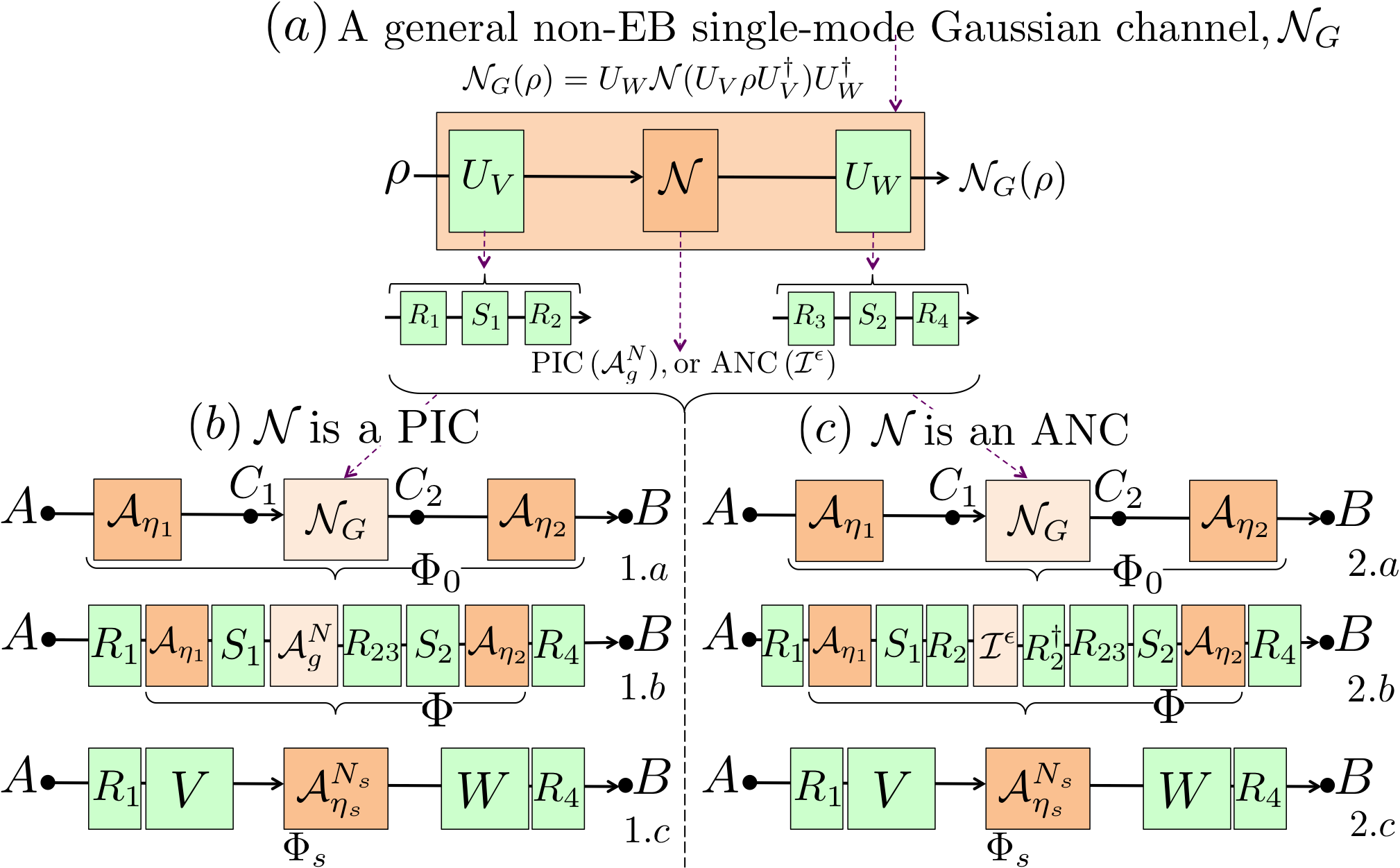}
\caption{(Color online) A general single-mode non-EB Gaussian channel ${\cal N}_G$ is unitary-equivalent to a Gaussian channel ${\cal N}$, which can be  one of two forms, a {\em phase-insensitive channel} (PIC), ${\cal A}_g^N$ , or a phase-sensitive {\em additive noise channel} (ANC), ${\cal I}^\epsilon$.　 For both cases of the center stations sandwiched between two lossy segmenets ${\cal A}_{\eta_1}$ and ${\cal A}_{\eta_2}$,  the total channel action $\Phi_0 \equiv {\cal A}_{\eta_2} \circ {\cal N}_G \circ {\cal A}_{\eta_1}$  is shown to be unitary-equivalent to a PIC channel  ${\cal A}_{\eta_s}^{N_s}$ as in (b)1.c and (c)2.c. 
Green-shaded boxes denote single-mode unitary (reversible) operations, whereas red-shaded boxes denote (in-general irreversible) actions of a single-mode quantum channel---a trace-preserving completely-positive map.}
\label{fig:onemodechannel_main}
\end{figure*}


There is no point in considering EB center stations ${\cal N}_G$ 
as they would trivially render $\Phi_0$ EB.  Any  single-mode Gaussian non-EB station   ${\cal N}_G$  is 
  unitary-equivalent to  either a phase insensitive channel (PIC) or an additive noise channel (ANC)   (See Sec.~\ref{sec:classification}).  In order to evaluate EB conditions, we go deeper into decomposing $\Phi_0$, as depicted in Fig.~\ref{fig:onemodechannel_main}.  The two branches of Fig.~\ref{fig:onemodechannel_main} consider decompositions  when ${\cal N}_G$ is unitary-equivalent to a  PIC  or an ANC, respectively. 

\noindent {\em 1. ${\cal N}_G$ unitary equivalent to a PIC ${\cal A}_g^N$}---Since phase-rotations commute with PICs, it is straightforward to see that the concatenated channel $\Phi_0 \equiv {\cal A}_{\eta_2} \circ {\cal N}_G \circ {\cal A}_{\eta_1}$ is unitary-equivalent to a channel $\Phi \equiv \mathcal A_{\eta_2} \circ \mathcal S_{ G_2, \theta  }\circ \mathcal A_{g}^{N}  \circ \mathcal S_{G_1} \circ \mathcal A_{\eta_1}$ [see Fig.~\ref{fig:onemodechannel_main}(b), lines~1.a~and~1.b], where $\mathcal S_{G_1}$ denotes a phase-quadrature squeezer (PSA) with gain $G_1$, and $\mathcal S_{G_2}$ is another PSA with gain $G_2$ and squeezing angle $\theta$. It is easy to deduce the parameters for the channel $\Phi$ as
\begin{align}
K_{\rm PIC}  =& \sqrt{g\eta _1 \eta_2} S_1 {R_\theta S_2  R_\theta^\dagger} \; ,\\
\alpha_{\rm PIC}  =& \frac{1}{2}\left[\eta_2(1-\eta_1) g K_\theta^T S_1 ^2 K_\theta + \right. \nonumber\\
 &\left. \eta_2(|g-1|+2N) K_\theta^T K_\theta + (1-\eta_2)  \openone_2\right],  \label{eq3333}
\end{align}
 respectively, where $K_\theta = R_\theta S_2  R_\theta^\dagger$, and
\begin{eqnarray}
S_i &=& \left(
  \begin{array}{cc}
    \sqrt G_i +\sqrt{G_i-1}   & 0   \\
       0 &   \sqrt G_i -\sqrt{G_i-1}   \\
  \end{array} \label{Gainnnnn}
\right) , \\
R_\theta &=& \left(
  \begin{array}{cc}
 \cos \theta   &  \sin \theta    \\
  -\sin\theta  & \cos \theta   \\
  \end{array}
\right).
\end{eqnarray}  
The following theorem shows  that the total unitary equivalent channel $\Phi$ is further unitary-equivalent to a PIC $\Phi_s$, as shown in line 1.c of Fig.~\ref{fig:onemodechannel_main}(b).

\begin{theorem}\label{thm:phi_phis}
$\Phi$ is unitary-equivalent to a channel $\Phi_s$ that is a PIC ${\cal A}_{\eta_s}^{N_s}$, whose descriptive parameters $(K_s,\alpha_s)$ are given by 
\begin{eqnarray}
K_s&=& V K W = \sqrt{g \eta_1 \eta_2 } \openone_2, \, {\text{and}}\nonumber \\
\alpha _s &=& W^T \alpha_{\rm PIC} W = \sqrt{\det(\alpha_{\rm PIC})} \openone_2,   \label{st0_main}
\end{eqnarray} 
where $V$ and $W$ are Gaussian unitaries.
\end{theorem}
\begin{proof}
See Appendix~\ref{app:proof_thm:phi_phis}.
\end{proof}

Comparing Eq.~\eqref{st0_main} with Eq.~\eqref{i}, it is easy to see that $\Phi_s$ is in fact a PIC, ${\cal A}_{\eta_s}^{N_s}$, with
\begin{eqnarray}
\eta_s  &=& g \eta_1 \eta_2, \,{\text{and}} \nonumber \\
N_s  &=& \sqrt{ \det (\alpha_{\rm PIC})} - \frac{|1- \eta_s|}{ 2}. \label{DefEff}
\end{eqnarray}

The condition under which the Gaussian center station ${\cal N}_G $ causes the lossy channel to be an EB channel is determined by applying Eq.~(\ref{38}) to the parameters of $\Phi_s$ in Eq.~\eqref{DefEff} since it is unitary equivalent to a PIC. Therefore, the channel $\Phi_0$ is EB  if
\begin{eqnarray}
\sqrt{\det(\alpha_{\rm PIC} )} \ge \frac{1}{2}(1+ g\eta_1\eta_2). \label{ebpic38}
\end{eqnarray}
\\

\noindent {\em 2. ${\cal N}_G$ unitary equivalent to an ANC ${\cal I}^\epsilon$}--- it is straightforward to deduce [see line 2.b in Fig.~\ref{fig:onemodechannel_main}(c)] that, $\Phi_0 \equiv {\cal A}_{\eta_2} \circ {\cal N}_G \circ {\cal A}_{\eta_1}$ is unitary-equivalent to a channel $\Phi = \mathcal A_{\eta_2} \circ \mathcal S_{ G_2, \theta }\circ (\mathcal R_{-\phi} \circ \mathcal I ^{\epsilon } \circ \mathcal R_\phi)   \circ \mathcal S_{G_1} \circ \mathcal A_{\eta_1}$, whose parameters are given by 

\begin{eqnarray}
K_{\rm ANC} &=& \sqrt{ \eta _1 \eta_2} S_1  K_\theta, \,{\text{and}}\\
\alpha_{\rm ANC}  &=& \frac{1}{2}\left[ \eta_2(1-\eta_1)  K_\theta^T S_1 ^2 K_\theta + \eta_2 K_\theta^T \epsilon^\prime K_\theta \right. \nonumber \\
&&\left. + (1-\eta_2)  \openone_2\right],
\end{eqnarray} 
where ${\cal R}_\phi$ is a phase rotation, $K_\theta = R_\theta S_2  R_\theta^\dagger$, and
\begin{eqnarray}
 \epsilon ':= \frac{1}{2}R_{\phi} \left(
  \begin{array}{cc}
    0   &  0  \\
     0  &  \epsilon   \\
  \end{array}
\right)
 R_{-\phi}^\dagger.
 \end{eqnarray}
Next we prove that $\Phi$ is unitary-equivalent to a PIC $\Phi_s$ [see line 2.c of Fig.~\ref{fig:onemodechannel_main}(b)].
\begin{theorem}\label{thm:phi_anc}
$\Phi$ is unitary-equivalent to a PIC $\Phi_s$ described by,
\begin{eqnarray}
K_s&=&  \sqrt{\eta_1 \eta_2 } \openone_2  , {\text{and}}\nonumber \\
\alpha _s &=& \sqrt{\det(\alpha_{\rm ANC})} \openone_2. \label{st1}
\end{eqnarray}
The excess  noise parameter is  given by 
\begin{eqnarray}
N_s  &=& \sqrt{ \det (\alpha_{\rm ANC})} - \frac{|1- \eta_1\eta_2 |}{ 2}. \label{ancns1}
\end{eqnarray}
\end{theorem}
\begin{proof}
As in the proof of Theorem~\ref{thm:phi_phis} (see Appendix~\ref{app:proof_thm:phi_phis}), we can simultaneously diagonalize $K$ and $\alpha$. This follows Eq.~\eqref{st1}.  Comparing Eqs.~\eqref{st1} and~\eqref{i}, we can determine the parameter $N_s$ in  Eq.~\eqref{ancns1}. 
\end{proof}
The condition under which the Gaussian center channel ${\cal N}_G $ causes the lossy channel to be an EB channel  
 is determined by applying Eq. (\ref{38}) to the parameters of $\Phi_s$ in Eqs.~\eqref{st1} and~\eqref{ancns1}  since it is unitary equivalent to a PIC. The condition for $\Phi_0$ to be EB, translates to 
\begin{equation}
\sqrt{\det(\alpha_{\rm ANC} )} \ge \frac{1}{2}(1+ \eta_1\eta_2).
\end{equation}

\subsection{Explicit examples of a Gaussian center station: optical amplifiers}\label{sec:PIAPSA}
In this subsection we illustrate our results in Section~\ref{sec:general_nonEB} with the important example of optical amplifiers used as center stations. We will consider the cases of a phase-sensitive amplifier (PSA) and phase-insensitive amplifier (PIA). Detailed proofs of the results will be deferred to Appendix~\ref{sec:PSAPIA}. If the center station ${\cal N}_G$ is a PSA of gain $G_{\rm PSA}$, then the composition $\Phi_0 \equiv {\cal A}_{\eta_2} \circ {\cal N}_G \circ {\cal A}_{\eta_1}$ becomes EB if the gain $G_{\rm PSA} $ exceeds a threshold $G^{\rm thres}_{\rm PSA}$  as
\begin{equation}
G_{\rm PSA} \ge G^{\rm thres}_{\rm PSA} := 1 + \frac{\eta_1}{(1-\eta_1)(1-\eta_2)}.
\end{equation} 
(See Appendix~\ref{sec:PSA_app} for proof.) If ${\cal N}_G$ is a PIA of gain $G_{\rm PIA}$, then $\Phi_0$ becomes EB if the Gain $G_{\rm PIA}$ exceeds a threshold value $G^{\rm thres}_{\rm PIA}$  (see Appendix~\ref{sec:PIA_app} for proof):
\begin{equation}
G_{\rm PIA} \ge G^{\rm thres}_{\rm PIA} := \frac{1}{1-\eta_1}.
\end{equation} 
Note that the transmittance $\eta_2$ of the loss segment {\em after} the PIA does not play a role in determining when $\Phi_0$ becomes EB. The expression for the threshold shows that when the channel transmittance  $\eta_1$ of a the initial is low, an amplifier with even a small amount of gain can render the lossy channel EB. Finally, the concatenation of a chain of PSA center stations, $\eta_1 \to {\rm PSA}(G_1) \to  \eta_2 \to {\rm PSA}(G_2) \to \ldots \to {\rm PSA}(G_k) \to \eta_{k+1}$, can be decomposed
as ${\cal N}_G^2 \circ {\cal A}_{\eta_1 \eta_2 \ldots \eta_{k+1}} \circ {\cal N}_G^1$, where ${\cal N}_G^1$ is a PSA at the channel input (of an appropriate gain and squeezing angle) followed by classical thermal noise addition, ${\cal A}_{\eta_1\eta_2\ldots\eta_{k+1}}$ is the entire channel loss collected in the middle, and ${\cal N}_G^2$ is a PSA at the channel output. For expressions of the gain and phase parameters of the PSAs at the transmitter and the receiver, see Appendix~\ref{sec:PSAchain}.

We note here that PIAs can improve the signal-to-noise ratio (SNR) of a sub-unity-efficiency optical heterodyne detection receiver, albeit up to 3 dB of the quantum limited SNR, when preceding the receiver. PSAs on the other hand have been proposed for use in optical imaging~\cite{Dut10} and secure-key generation~\cite{Zha14}, to boost the effective detection efficiency of homodyne detection receivers, in principle pushing the receiver's performance all the way to the quantum limited SNR, by preceding the receiver with a PSA whose gain quadrature is phase-matched to the homodyne detector's local oscillator. Despite these practical uses of optical amplifiers, 
our results in the earlier sections show that these amplifiers cannot increase the secret key capacity, and the results in the current section show that it is unlikely that they will help to realize the given secret key capacity in a  practical implementation.

\section{Conclusions}\label{sec:conclusions}
It was recently shown \cite{Tak13} that for QKD (secure key generation), along with a few other optical quantum communication tasks such as quantum (qubit) communication, entanglement generation, and direct-secure communication (each with two-way authenticated classical communication assistance), the rates are upper bounded by $R_{\rm UB} = \log_2[(1+\eta)/(1-\eta)]$ bits per mode over a pure-loss optical channel of transmittance $\eta$. This upper bound reads $R_{\rm UB} \approx 2.88\eta$ when $\eta \ll 1$ (high loss), which translates to an exponential decay of rate with distance $L$ in fiber ($\eta \propto e^{-\alpha L}$), and an inverse-square decay with $L$ in free-space ($\eta \propto 1/L^2$). Quantum repeaters are conceptual devices that,  when inserted along the lossy channel, can help circumvent this rate-loss trade-off. 

In this paper, we have proven the inefficacy of  bosonic Gaussian channels---optical processes that can be assembled using passive linear optics (beamsplitters and phase-shifters) and squeezers (phase-sensitive amplifiers, and the interaction  of parametric downconversion)---to be used as quantum repeaters. We prove this by showing that any concatenation of such untended Gaussian operations along a lossy channel can be simulated by one Gaussian operation at the channel input and one at the channel output, where the entire loss in the channel is collected in the middle. We thereby argue that any communication protocol that uses such a chain of Gaussian center stations can be replaced by another protocol of the same performance without those stations, the transmitter and receiver of which are slightly modified versions of those used by the original protocol. As a consequence, the upper bound $R_{\rm UB}$, as shown above, still applies.  
Note, however,  that our formulation is entirely based on the property of Gaussian channels and does not preclude the possibility that a trace-decreasing Gaussian operation \cite{Giedke2002}  could serve  as a quantum repeater. 


  It would be possible that intermediate trace-preserving Gaussian operations  could be of practical advantage, while the same performance of any protocol working with such middle stations is  in principle  achievable {\it without} middle  stations. 
  In order to demonstrate  practical restrictions for use of  conventional Gaussian stations, we separately analyzed the case of a general  single-mode  Gaussian channel  sandwiched  between  lossy channels. We derived the conditions that the center station renders the end to end lossy channel entanglement breaking, and hence useless for QKD.  From special cases for quantum-limited optical amplifiers as center stations,  we found that in a high-loss regime, even modest amplification gains will render the overall channel entanglement breaking. 
  

\acknowledgments
 RN, OG, and NL were supported by the DARPA Quiness program  under prime contract number W31P4Q-12-1-0017. SG was supported by the DARPA Quiness program subaward contract number SP0020412-PROJ0005188, under prime contract number W31P4Q-13-1-0004.

\newpage

\appendix
\section{Extracting one mean displacement in a concatenated Gaussian operation}\label{app:displacement}

The action of the concatenation of $n$ Gaussian channels ${\cal E}_n \circ \ldots \circ {\cal E}_2 \circ {\cal E}_1$, where ${\cal E}_i \triangleq (K_i, m_i, \alpha_i)$ can always be mimicked by a concatenation ${\cal E}^\prime_n \circ \ldots \circ {\cal E}^\prime_2 \circ {\cal E}^\prime_1$, where all the displacement terms are pushed to the $n$-th channel, i.e., ${\cal E}^\prime_i \triangleq (K_i, 0, \alpha_i)$, $1 \le i \le n-1$, and ${\cal E}^\prime_n \triangleq (K_n, m_t, \alpha_n)$, where $m_t$ is a function of $\left\{m_1, \ldots, m_n\right\}$, and $\left\{K_1, \ldots, K_n\right\}$. To see this, consider the composition of $n$ Gaussian channels:
\begin{eqnarray}
\mathcal E_{123 ... n } &=&  \mathcal E_n \circ \cdots \circ \mathcal E_3 \circ \mathcal E_2 \circ \mathcal E_1, 
\end{eqnarray}
for which we may write,
\begin{eqnarray}
d_{123 ... n }  &:=&  K_{123...n}^T d + m_{12...n} \nonumber, \\
      &=& K_n^T K_{123...n-1}^T  d +  K_n ^T m_{123 ...n-1} +m_n   \nonumber \\
     & \vdots& \nonumber \\
     &=& K_{1 \to n}^T d + \sum_{j=2}^n K_{j \to n }^T m_{j -1} +m_n \nonumber \\   &= & K_t d + m_t \label{70}%
\end{eqnarray}
where
\begin{eqnarray}
 K_t :&=&   K_{1  \to n } \; ,\nonumber \\  
  m_t:&=&  \sum_{j=2}^n K_{j \to n }^T m_{j -1} +m_n \; ,\\
 K_{j \to n } :&=&  \begin{cases}
  K_j K_{j+1} K_{j+2} \cdots K_{n-1} K_n   &  j< n-1   \nonumber \\
          K_n  & j=n. 
\end{cases}
\end{eqnarray}
From Eq. (\ref{70}) we can confirm that the total change in the first moment $d$ is given  by $K_t =  K_{1 \to n}  $ and a constant shift $m_t= \sum_{j=2}^n K_{j \to n }^T m_{j -1} +m_n $. Hence, we can obtain the same transformation of $d$, for example, by setting $m_1 = m_2 = \cdots = m_{n-1} =0 $ and $m_n = m_t$, while leaving the gain terms $K_j$ of each channel $\mathcal E_ j $ as they are. To be precise, the following two channels equivalently act on $(\gamma, d) $.
\begin{eqnarray}
E & := & \underbrace{\mathcal E_n}_{(K_n, m_n,\alpha_n)}   \circ    \cdots  \circ  \underbrace{\mathcal E_2}_{(K_2, m_2,\alpha_2)} \circ  \underbrace{\mathcal E_1}_{(K_1, m_1,\alpha_1)},    
 \\ 
E' & : =   & \underbrace{\mathcal E_n}_{(K_n, m_t,\alpha_n)}   \circ    \cdots  \circ  \underbrace{\mathcal E_2}_{(K_2, 0,\alpha_2)} \circ  \underbrace{\mathcal E_1}_{(K_1, 0,\alpha_1)}.  \end{eqnarray}
Therefore, mean  displacement terms can be absorbed into the final Gaussian operation, and their  effect  can be treated  separately in the analysis of the sequential channel action.  In this manner,    
 one can usually discuss Gaussian channel properties by assuming $m_j=0$ for all $j$ without loss of generality, and taking into account the effect of $m$'s, if at all needed, at once.

\section{Proof of Theorem~\ref{thm:phi_phis}}\label{app:proof_thm:phi_phis}

We restate Theorem~\ref{thm:phi_phis} below for completeness:

\begin{theorem*}
$\Phi$ is unitary-equivalent to a channel $\Phi_s$ that is a PIC ${\cal A}_{\eta_s}^{N_s}$, whose descriptive parameters $(K_s,\alpha_s)$ are given by
\begin{eqnarray}
K_s&=& V K W = \sqrt{g \eta_1 \eta_2 } \openone_2, \, {\text{and}}\nonumber \\
\alpha _s &=& W^T \alpha W = \sqrt{\det(\alpha)} \openone_2,   \label{st0}
\end{eqnarray} 
where $V$ and $W$ are Gaussian unitaries.
\end{theorem*}
\begin{proof}
Let $W_0$ be an orthogonal matrix that diagonalizes $\alpha$ so that $W_0^T \alpha W_0= \textrm{diag} (\lambda_1, \lambda_2) $. We then have, $\det(\alpha ) = \lambda_1 \lambda_2 $. It is then easy to see that the expression for $\alpha_s$ in Eq. (\ref{st0}) can be obtained by choosing 
\begin{eqnarray}
W = (\lambda_1 \lambda_2)^{-1/4}{W_0 \sqrt{\textrm{diag} (\lambda_2, \lambda_1) }} .
\end{eqnarray}
Given this $W$, we can choose $V= W^{-1} K_\theta^{-1} S_1^{-1}$ to obtain the expression for $K_s$ in Eq. (\ref{st0}). Decomposing $\Phi$ into $\Phi_s$, sandwiched between unitaries $V$ and $W$ is depicted in line 1.c of Fig.~\ref{fig:onemodechannel_main}(b).
\end{proof}

 

\section{Analysis of optical amplifiers as regenerative stations}\label{sec:PSAPIA}

In this Appendix, we prove the entanglement breaking conditions stated in Section~\ref{sec:PIAPSA}, for when an optical amplifier, either phase-insensitive (PIA) or phase-sensitive (PSA), is used as a center station, sandwiched between two pure-loss channel segments. We will evaluate these conditions by applying our general results from Section~\ref{sec:general_nonEB}.

The decomposition shown in line 1.b of Fig.~\ref{fig:onemodechannel_main}(b), with the excess noise parameter $N$ of ${\cal A}_g^N$ set to zero, includes the quantum-noise-limited PSA and PIA as special cases. For $g \ge 1$ and $N=0$, we obtain a simple expression of the determinant of $\alpha$ in Eq.~\eqref{eq3333}:
\begin{eqnarray}
\det(\alpha) &=&\frac{1}{4}\bigg \{ (1-g \eta_1\eta_2)^2
 -4 \eta_2   \bigg[G_2 (1 - g \eta_1) (1 - \eta_2)    \nonumber \\  && 
  +  g G_1 (1 - \eta_1) (1  - g \eta_2) \nonumber \\
  && - 2 g\sqrt{ {G_1} {G_2}}(1 - \eta_1) (1 - \eta_2) \nonumber \\
  &&   \times \left( \sqrt{ {G_1} {G_2}}
  +    \sqrt{( G_1-1)(G_2-1) } \cos 2\theta \right)\bigg] \bigg\}.  \nonumber \\ \label{deta} \end{eqnarray}

\subsection{PSA sandwiched by two lossy channels}\label{sec:PSA_app}
For the case of the quantum-limited PSA, by setting ${G_2=g= 1}$ in Eq. (\ref{deta}) one obtains 
\begin{align}
&  \det( \alpha) =  \frac{   (1- \eta_1 \eta_2)^2}{4} +   \eta_2 (G_1-1) (1 - \eta_1) (1 -\eta_2) . 
\end{align}
From this relation and with $\eta_s = \eta_1 \eta_2$ the EB condition of Eq.~\eqref{ebpic38}  reads 
\begin{eqnarray}
 G_1 \ge 1+ \frac{ \eta_1}{(1- \eta_1)(1- \eta _2)}.  \label{interest}
\end{eqnarray}

\begin{remark}
A pure loss channel ${\cal A}_\eta$, $\eta \le 1$, is not EB. A quantum-limited PSA (which is a squeezer, and hence a unitary) is not EB. This is consistent with the observation that setting either $\eta_1$ or $\eta_2$ close to $1$ requires the PSA gain $G_1$ to go to infinity in order for the composition (loss-PSA-loss) to be EB. It is interesting that even though pure-loss channels and quantum-limited PSA are not EB by themselves, composing them can yield an EB channel if the gain and transmittances satisfy the condition in Eq.~\eqref{interest}.
\end{remark}

\subsection{PIA sandwiched by two lossy channels}\label{sec:PIA_app}

For the case of the quantum-limited PIA by setting ${G_1=G_2= 1}$ in Eq. (\ref{deta}) one obtains 
\begin{eqnarray}
\det (\alpha)= \frac{(1 + 2 (g-1)\eta_2 - g \eta_1  \eta_2 )^2}{4}.
\end{eqnarray}
From this relation 
and with $\eta_s = g \eta_1 \eta_2$,   the EB condition of  Eq.~\eqref{ebpic38}    now reads
\begin{eqnarray}\label{eq:EB_PIA}
 g \ge         \frac{1}{1-\eta_1}. \end{eqnarray}
\begin{remark}
One notable point is that the transmittance $\eta_2$ of the lossy channel that appears {\em after} the PIA, does not play a role in the EB condition in Eq.~\eqref{eq:EB_PIA}.
\end{remark}

\subsection{Analysis of a chain of PSA center stations}\label{sec:PSAchain}
Let us consider  a  chain of PSA center stations, interspersed between the number of $k+1$ lossy segments with transmission  $\{\eta_i\}_{i= 1,2, \cdots,  k+1}$ as in Fig.~\ref{fig:psaseq}(a). The channel action is formally written as 
\begin{align}
\Phi_0  =& \mathcal A_{\eta_{k+1}} \circ \mathcal S_{ G_k } \circ \mathcal A_{\eta_k}  \circ  \cdots \nonumber \\
 &\cdots \circ \mathcal A_{\eta_3} \circ \mathcal S_{ G_2 } \circ  \mathcal A_{\eta_2} \circ \mathcal S_{G_1} \circ \mathcal A_{\eta_1},  \label{eqpsas}
\end{align}  
where  the action of PSAs $ \mathcal S_{ G_i }  $  with the amplification gain of $\{G_i\}_{i= 1,2, \cdots, k}$ can be described   by  Eq.~\eqref{Gainnnnn}. 

\begin{figure}[tphb]
\centering
\includegraphics[width=\columnwidth]{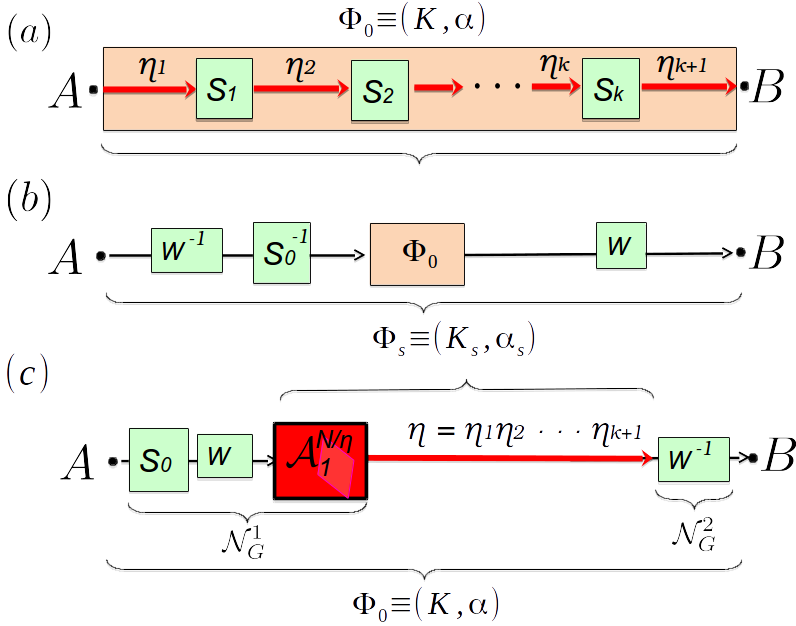}
\caption{(Color online)   (a) A sequence of PSAs $S_i$ with $i= 1,2, \cdots, k$  connected by the lossy segments of  transmission $\eta_i$ with  $i= 1, 2, \cdots , k+1$. The total loss is given by $\eta = \prod_{i=1}^{k+1} \eta_i$.  (b) The PSA-loss chain $\Phi_0$ can be transformed to the standard form of a PIC $\Phi_s $ by  using squeezing unitary operations $W$ and $S_0$.  (c) The standard form $\Phi_s$ is decomposed into a thermal noise channel ${\cal A}_1^{N/\eta}$ and a transmission-$\eta $ pure lossy segment.  Then,   the origin channel $\Phi_0$ can be simulated by  adding unitary operators to cancel out  the unitary operators in (b) at the input-end and output end. This turns $\Phi_0$ into the form with original loss sandwiched by the operation ${\cal N}_G^1$  and the optput-end operation ${\cal N}_G^1$,  ---  an explicit example of our main result explained  in Fig.~\ref{fig:summary}.}
\label{fig:psaseq}
\end{figure}


By repeatedly using the composition rule of Eq.~(\ref{CompoRule}) we can write the channel parameters  $\Phi_0  \triangleq (K, 0, \alpha)$ as follows:
\begin{widetext}
\begin{align}
K&= \sqrt{ \eta_{k+1}\eta_{k} \eta_{k-1}\cdots \eta_{1}} S_1 S_2 \cdots S_k = \sqrt{\eta} S_0, & \\
\alpha& = \frac{1}{2}\left[ \bar \eta_{k+1} \openone_2 + \eta_{k+1} \left\{ \bar \eta_k S_k^T S_k + \eta _k \bar \eta_{k-1} S_{k}^T S_{k-1}^T S_{k-1}S_k + \eta _k  \eta _{k-1} \bar \eta_{k-2} S_{k}^T S_{k-1}^T S_{k-2}^T S_{k-2} S_{k-1}S_k +\cdots  \right\}  \right] \nonumber \\
&=\frac{1}{2}\left[ \bar \eta_{k+1} \openone_2 + \eta_{k+1} \left\{\sum_{n=0}^{k-1}  \frac{\bar \eta_{k-n}}{\eta_{k-n}} \prod_{j=0}^{n} \eta_{k-j} (S_{k-j})^2  \right\}  \right] = \left(
  \begin{array}{cc}
    \alpha^{(+)}   &  0  \\
     0  & \alpha^{(-)}  \\
  \end{array}
\right),
\end{align} where
\begin{align}
{\eta} &:=  { \eta_{k+1}\eta_{k} \eta_{k-1}\cdots \eta_{1}}  , & \\
\bar \eta_i &:=  1- \eta_i ,\\
  S_0&:=  S_1 S_2 \cdots S_k ,   & \\
\alpha^{(\pm)}& :=  \frac{1}{2}\left[ \bar \eta_{k+1}  + \eta_{k+1} \left\{\sum_{n=0}^{k-1}  \frac{\bar \eta_{k-n}}{\eta_{k-n}} \prod_{j=0}^{n} \eta_{k-j}\left (\sqrt{G_{k-j} } \pm  \sqrt{G_{k-j}-1}\right )^2  \right\}  \right].
\end{align}
\end{widetext}

Let us define a squeezer 
\begin{align}
W := \left(\alpha^{(+)} \alpha^{(-)} \right )^{-1/2} \sqrt{\textrm{diag}[ \alpha^{(-)}, \alpha^{(+)}]  } \end{align}
that symmetrizes $\alpha $ as  $W^T \alpha W  \propto  \openone_2$ and set  $V= W^{-1}S_0^{-1}$   similarly to the proof of Theorem~\ref{thm:phi_phis} in Appendix~\ref{app:proof_thm:phi_phis}  [See Fig.~\ref{fig:psaseq}(b)]. Then, we can convert $\Phi_0$ to the standard form of a PIC $\Phi _S  \triangleq (K_s,0, \alpha_s ) $ with 
\begin{align}
K_s &= \sqrt{\eta} \openone_2,\\
\alpha_s&= \sqrt{ \det(\alpha ) } \openone_2 = \left[ \frac{1}{2} (1-\eta ) +N \right] \openone_2.  
\end{align}  
This  implies the EB condition due to Eq.~\eqref{38}: 
 \begin{align}
 \sqrt{ \det(\alpha ) }  \ge \frac{1}{2}   (1+\eta ).
\end{align} 

Let us now explicitly show the decomposition of the PSA chain into a pair of Gaussian operations at the input and the output. See Fig.~\ref{fig:psaseq} for a pictorial depiction. From the standard form we can split out the pure lossy segment  by using the relation $\Phi_s\equiv {\cal A}_\eta^{N}  = {\cal A}_\eta \circ {\cal A }_1^{N/\eta }$, which can be confirmed easily  from  the composition rule of Eq.~(\ref{CompoRule}). We can retrieve the original channel $\Phi_0$ by canceling  the unitary operators $W$, $S_0$, and $W^{-1}$ as in Fig.~\ref{fig:psaseq}(c). To be specific, we can write  the original  channel in the sandwiched form $\Phi _0 = {\cal N }_G^{2} \circ {\cal A}_ \eta \circ {\cal N }_G^{1} $ with the two Gaussian channels ${\cal N }_G^{1}   =  {\cal A}_{1}^{N/ \eta} \circ {\cal W} \circ {\cal S}_0^{-1}$ and ${\cal N }_G^{2}   =   {\cal W } ^{-1} $. Note that the channel at the receiver ${\cal N }_G^{2}$ is  a squeezing  unitary, which is  consistent with  Proposition~\ref{prop:n-mode}.


\begin{thebibliography}{}
\bibitem{footnote}{For an optical fiber channel, transmittance decreases exponentially with distance $L$, i.e., $\eta \propto e^{-\gamma L}$, where $\gamma$ is the fiber loss coefficient that depends upon the transmission center-wavelength $\lambda$. For a single-mode free-space channel in the far-field regime, $\eta \approx \left[A_tA_r/(\lambda L)^2\right]e^{-\alpha L}$, where $A_t$ and $A_r$ are areas of the transmitter and receiver apertures, and $\alpha$ is the atmospheric extinction coefficient.}
\end{thebibliography}
\end{document}